\documentclass[a4paper]{amsart}

\usepackage{amsmath}
\usepackage{amssymb}
\usepackage{amsmath}
\usepackage{mathtools} 

\theoremstyle{definition}
\newtheorem{definition}{Definition}

\theoremstyle{theorem}
\newtheorem{proposition}{Proposition}
\newtheorem{lemma}{Lemma}
\newtheorem{theorem}{Theorem}

\usepackage{tikz-cd}
\usepackage{hyperref}

\newcommand{\parfunc}{\rightharpoonup}
\newcommand{\natexp}{\overset{\bullet}{\looparrowright}}
\newcommand{\bars}[1]{\left\lvert #1 \right\rvert}
\newcommand{\dbars}[1]{\left\| #1 \right\|}
\newcommand{\setcomp}[2]{\left\{ \, #1 \; \middle| \; #2 \, \right\}} 
\newcommand{\myeq}{\stackrel{\mathclap{\tiny\mbox{def}}}{=}}
\newcommand{\mb}[1]{\mathbf{#1}}
\newcommand{\expo}{\looparrowright}

\newcommand{\ibox}[1]{\mathsf{box\;}#1}

\title{On the Semantics of Intensionality}
\thanks{This work was supported by the EPSRC (award reference 1354534).
The final publication is available at Springer via
\url{http://dx.doi.org/10.1007/978-3-662-54458-7_32}.}

\author{G. A. Kavvos}

\address{Department of Computer Science, University of Oxford
 Wolfson Building, Parks Road, Oxford OX1 3QD, United Kingdom}
\email{alex.kavvos@cs.ox.ac.uk}

\begin{document}

\maketitle

\begin{abstract}

In this paper we propose a categorical theory of intensionality.
We first revisit the notion of intensionality, and discuss its
relevance to logic and computer science. It turns out that
1-category theory is not the most appropriate vehicle for studying
the interplay of extension and intension. We are thus led to
consider the P-categories of \v{C}ubri\'{c}, Dybjer and Scott,
which are categories only up to a partial equivalence relation
(PER). In this setting, we introduce a new P-categorical
construct, that of exposures. Exposures are very nearly functors,
except that they do not preserve the PERs of the P-category.
Inspired by the categorical semantics of modal logic, we begin to
develop their theory. Our leading examples demonstrate that an
exposure is an abstraction of well-behaved intensional devices,
such as G\"odel numberings. The outcome is a unifying framework in
which classic results of Kleene, G\"odel, Tarski and Rice find
concise, clear formulations, and where each logical device or
assumption involved in their proofs can be expressed in the same
algebraic manner.

\end{abstract}

\section{Introduction: Intensionality \& Intensional Recursion}

This paper proposes a new theory of \emph{intensionality}.
Intensionality is a notion that dates as early as Frege's
philosophical distinction between \emph{sense} and
\emph{denotation}, see for example \cite{Fitting2015}. In the most
mathematically general sense, to be `intensional' is to somehow
operate at a level finer than some predetermined `extensional'
equality. Whereas in mainstream mathematics intensionality is
merely a nuisance, it is omnipresent in computer science, where
the objects of study are distinct programs and processes
describing (often identical) abstract mathematical values.  At one
end of the spectrum, programs are extensionally equal if they are
\emph{observationally equivalent}, i.e.  interchangeable in any
context. At the other extreme, computer viruses often make
decisions simply on patterns of object code they encounter,
disregarding the actual function of what they are infecting; one
could say they operate up to syntactic identity.

\subsection{Intensionality as a Logical Construct}
  \label{sec:intlog}

We are interested in devising a categorical setting in which
programs can be viewed in two ways simultaneously, either as
\emph{black boxes}---i.e.  extensionally, whatever we define that
to mean, but also as \emph{white boxes}---i.e.  intensionally,
which should amount to being able to `look inside' a construction
and examine its internal workings.

There are many reasons for pursuing this avenue. The main new
construct we will introduce will be an abstraction of the notion
of \emph{G\"odel numbering}. The immediate achievement of this
paper is a categorical language in which we can state many classic
theorems from logic and computability that depend on the
interplay between extension and intension. This unifying language
encompasses all such `diagonal constructions' in a way that makes
the ingredients involved in each argument clear. As such, we
regard this as an improvement on the classic paper of Lawvere
\cite{Lawvere2006}.

A more medium-term goal is the quest to prove a logical foundation
to \emph{computational reflection}, in the sense of Brian Cantwell
Smith \cite{Smith1984}. A reflective program is always able to
obtain a complete description of its source code and current
state; this allows it to make decisions depending on both its
syntax and runtime behaviour. This strand of research quickly ran
into impossibility results that demonstrate that reflective
features are logically ill-behaved: see e.g. \cite{Barendregt1991}
for reflection in untyped $\lambda$-calculus, or \cite{Wand1998}
for a more involved example involving the LISP $\textbf{fexpr}$
construct. Our viewpoint allows us to talk about the notion of
\emph{intensional recursion}, which is more general than ordinary
extensional recursion, and seems to correspond to a well-behaved
form of reflection. We connect this to a classic result in
computability theory, namely Kleene's \emph{Second Recursion
Theorem (SRT)}.

Finally, a more long-term goal is to understand
\emph{non-functional computation}. In this context, non-functional
computation means something much more general than just computing
with side-effects: we are interested in general higher-order
computation acting `on syntax.' Approaches to such forms of
computation have hitherto been ad-hoc, see e.g. \cite[\S
6]{Longley2000}. We would like to provide a very general way to
add `intensional' features to a pure functional programming
language.

\subsection{Prospectus}

To begin, we introduce in \S\ref{sec:modp} a known connection
between the notion of intension and the necessity modality
from modal logic, and the use of \emph{modal types} in isolating
intension from extension. We argue that this connection cannot be
fully substantiated in 1-category theory. Hence, we introduce
\emph{P-categories} and explain their use in modelling
intensionality.

In \S\ref{sec:expo} we introduce a new P-categorical construct,
the \emph{exposure}. Exposures `turn intensions into
extensions' in a manner inspired by the modality-as-intension
interpretation. In \S\ref{sec:intrec} we use exposures to talk
about the notion of intensional recursion in terms of
\emph{intensional fixed points (IFPs)}. 

In \S\ref{sec:consistency} we use IFPs to prove abstract
analogues to \emph{Tarski's undefinability theorem} and G\"odel's
\emph{First Incompleteness Theorem}. In \S\ref{sec:arith} we
construct a P-category and an endoexposure that substantiate the
claim that the the abstract versions correspond to the usual
theorems.

We then ask the obvious question: where do IFPs come from? In
\S\ref{sec:lawvere} we generalize Lawvere's fixed-point theorem to
yield IFPs. Then, in \S\ref{sec:recthm}, we draw a parallel
between Kleene's \emph{First Recursion Theorem (FRT)} and
Lawvere's fixed point result, whereas we connect our IFP-yielding
result with Kleene's \emph{Second Recursion Theorem}. We
substantiate this claim in \S\ref{sec:asm} by constructing a
P-category and exposure based on realizability theory.

Finally, in \S\ref{sec:rice}  we provide further evidence for the
usefulness of our language by reproducing an abstract version of
\emph{Rice's theorem}, a classic result in computability theory.
We find that this is already substantiated in the P-category
constructed in \S\ref{sec:asm}.

\section{Modality-as-Intension}
  \label{sec:modp}

All the negative results regarding intensionality and
computational reflection have something in common: they invariably
apply to some construct that can \emph{turn extension into
intension}. For example, in \cite{Barendregt1991} a contradiction
is derived from the assumption that some term $Q$ satisfies $QM
=_{\beta} \ulcorner M \urcorner$ for any $M$, where
the RHS is a G\"odel number of $M$. The moral is that \emph{one
should not mix intension and extension}.

To separate the two, we will use \emph{modal types}, as first
suggested by Davies and Pfenning \cite{Davies2001a,Davies2001}. In
\emph{op. cit.} the authors use modal types to simulate two-level
$\lambda$-calculi. In passing, they interpret the modal type $\Box
A$ as the type of `intensions of type $A$' or `code of type $A$.'
The \textsf{T} axiom $\Box A \rightarrow A$ may then be read as an
\emph{interpreter} that maps code to values, whereas the
\textsf{4} axiom $\Box A \rightarrow \Box\Box A$ corresponds to
\emph{quoting}, but quoting that can only happen when the initial
value is already code, and not a `runtime,' live value.

Unfortunately, the available semantics do not corroborate this
interpretation. The categorical semantics of the \textsf{S4}
necessity modality---due to Bierman and de Paiva
\cite{Bierman2000a}---specify that $\Box$ is a \emph{monoidal
comonad} $(\Box, \epsilon, \delta)$ on a CCC. The problem is now
rather obvious, in that equality in the category \emph{is}
extensional equality: if $f = g$, then $\Box f = \Box g$. In modal
type theory, this amounts to saying that $\vdash M = N : A$
implies $\vdash \ibox{M} = \ibox{N} : \Box A$, which is not what
we mean by intensionality at all. By definition, `intension'
should not be preserved under equality, and in \cite{Davies2001a}
it is clearly stated that there should be no reductions under a
$\ibox{(-)}$ construct.

To salvage this \emph{modality-as-intension} interpretation, we
have to leave 1-category theory, and move to a framework where we
can separate \emph{extensional equality}---denoted $\sim$---and
\emph{intensional equality}---denoted $\approx$. We will thus use
the \emph{P-categories} of \v{C}ubri\'{c}, Dybjer and Scott
\cite{Cubric1998}. P-categories are only categories up to a family
of partial equivalence relations (PERs). In our setting, the PER
will specify extensional equality. All that remains is to devise a
construct that (a) behaves like a modality, so that intension and
extension stay separate, and (b) unpacks the non-extensional
features of an arrow. This will be the starting point of our
theory of \emph{exposures}.

\subsection{P-categories and Intensionality}

Suppose we have a model of computation or a programming language
whose programs are seen as computing functions, and suppose that
we are able to \emph{compose} programs in this language, so that
given programs $\texttt{P}$ (computing $f$) and $\texttt{Q}$
(computing $g$) there is a simple syntactic construction
$\texttt{Q}; \texttt{P}$ (computing $g \circ f$). In more elegant
cases, like the $\lambda$-calculus, composition will be
substitution of a term for a free variable.  But in most other
cases there will be unappealing overhead, involving e.g.
some horrible disjoint unions of sets of states. This syntactic
overhead almost always ensures that composition of programs
is not associative: $(\texttt{R};\texttt{Q});\texttt{P}$ is
\emph{not} syntactically identical to
$\texttt{R};(\texttt{Q};\texttt{P})$, even though they compute the
same function.

To model this we will use \emph{P-categories}, first introduced in
\cite{Cubric1998}. Generally denoted $\mathfrak{B}, \mathfrak{C},
\dots$, or even $(\mathfrak{B}, \sim)$, P-categories are
categories whose hom-sets are not sets, but \emph{P-sets}: a P-set
is a pair $A = (\bars{A}, \sim_A)$ of a set $\bars{A}$ and a
partial equivalence relation (PER)\footnote{That is, a symmetric
and transitive relation.} on $\bars{A}$. If $a \sim_A a'$, then
$a$ can be thought of as `equal' to $a'$. The lack of reflexivity
means that there may be some $a \in \bars{A}$ such that $a
\not\sim_A a$: these can be thought of as points which are
\emph{not well-defined}. A \emph{P-function} $f : A \rightarrow B$
between two P-sets $A = (\bars{A}, \sim_A)$ and $B = (\bars{B},
\sim_B)$ is a function $\bars{f} : \bars{A} \rightarrow \bars{B}$
that respects the PER: if $a \sim_A a'$ then $\bars{f}(a) \sim_B
\bars{f}(a')$.  We simply write $f(a)$ if $a \sim_A a$.

Thus, we take each hom-set $\mathfrak{C}(A, B)$ of a P-category
$\mathfrak{C}$ to be a P-set. We will only write $f : A
\rightarrow B$ if $f \sim_{\mathfrak{C}(A, B)} f$, i.e. $f$ is a
well-defined arrow. Arrows in a P-category are \emph{intensional
constructions}. Two arrows $f, g : A \rightarrow B$ will be
\emph{extensionally equal} if $f \sim_{\mathfrak{C}(A, B)} g$. The
axioms of category theory will then only hold up to the family of
PERs, i.e. ${\sim} = \{\sim_{\mathfrak{C}(A, B)}\}_{A, B \in
\mathfrak{C}}$. For example, it may be that $f \circ (g \circ h)
\neq (f \circ g) \circ h$, yet $f \circ (g \circ h) \sim (f \circ
g) \circ h$.

Hence, we regain all the standard equations of 1-categories, up to
PERs. Furthermore, the standard notions of terminal objects,
products and exponentials all have a P-variant in which the
defining equations hold up to $\sim$. We are unable to expound on
P-categories any further, but please refer to \cite{Cubric1998}
for more details.

\section{Exposures}
  \label{sec:expo}

We can now formulate the definition of exposures.  An exposure is
\emph{almost} a (P-)functor: it preserves identity and
compositions, but it only \emph{reflects} PERs.

\begin{definition}
  An exposure $Q : (\mathfrak{B}, \sim) \expo{} (\mathfrak{C},
  \sim)$ consists of (a) an object $QA \in \mathfrak{C}$ for each
  object $A \in \mathfrak{B}$, and (b) an arrow $Qf : QA
  \rightarrow QB$ in $\mathfrak{C}$ for each arrow $f : A
  \rightarrow B$ in $\mathfrak{B}$, such that (1) $Q(id_A) \sim
  id_{QA}$, and (2) $Q(g \circ f) \sim Qg \circ Qf$ for any arrows
  $f : A \rightarrow B$ and $g : B \rightarrow C$, and (3) for any
  $f, g : A \rightarrow B$, if $Qf \sim Qg$ then $f \sim g$.
\end{definition}

The \emph{identity exposure} $\mathsf{Id}_\mathfrak{B} :
\mathfrak{B} \expo{} \mathfrak{B}$ maps every object to itself,
and every arrow to itself. Finally, it is easy to see that the
composite of two exposures is an exposure.

As exposures give a handle on the internal structure of arrows,
they can be used to define intensional equality: if the images of
two arrows under the same exposure $Q$ are extensionally equal,
then the arrows have the same implementation, so they are
intensionally equal. This is an exact interpretation of a slogan
of Abramsky \cite{Abramsky2014}: \emph{intensions become
extensions}.

\begin{definition}[Intensional Equality]
  Let there be P-categories $\mathcal{B}$, $\mathcal{C}$, and an
  exposure $Q : (\mathcal{B}, \sim) \expo{} (\mathcal{C}, \sim)$.
  Two arrows $f, g : A \rightarrow B$ are \emph{intensionally
  equal (up to $Q$)}, written $f \approx g$, just if $Qf \sim Qg$.
\end{definition}

It is obvious then that the last axiom on the definition of
exposures means that intensional equality implies extensional
equality.

To re-interpret concepts from the modality-as-intension
interpretation---such as interpreters, quoting etc.---we shall
need a notion of transformation between exposures. 

\begin{definition}
  A \emph{natural transformation of exposures} $t : F \natexp G$
  where $F, G: \mathfrak{B} \expo{} \mathfrak{C}$ are exposures,
  consists of an arrow $t_A : FA \rightarrow GA$ of $\mathfrak{C}$
  for each object $A \in \mathfrak{B}$, such that, for every arrow
  $f : A \rightarrow B$ of $\mathfrak{B}$, the following diagram
  commutes up to $\sim$: \[
    \begin{tikzcd}
      FA
        \arrow[r, "Ff"]
        \arrow[d, "t_A", swap]
      & FB
        \arrow[d, "t_B"] \\
      GA
        \arrow[r, "Gf", swap]
      & GB
    \end{tikzcd}
  \] 
\end{definition}

\subsection{Cartesian Exposures}

Bare exposures offer no promises or guarantees regarding
intensional equality. For example, it is not a given that $\pi_1
\circ \langle f, g \rangle \approx f$. However, one may argue that
this equality should hold, insofar as there is no grand
intensional content in projecting a component. This leads to the
following notion:

\begin{definition}
  A exposure $Q : \mathfrak{B} \expo{} \mathfrak{C}$ where
  $\mathfrak{B}$ is a cartesian P-category is itself
  \emph{cartesian} just if, for arrows $f : C \rightarrow A$
  and $g : C \rightarrow B$, we have \[
    \pi_1 \circ \langle f, g \rangle \approx f,
      \quad
    \pi_2 \circ \langle f, g \rangle \approx g,
      \quad\text{and }
    \langle \pi_1 \circ h, \pi_2 \circ h \rangle \approx h
  \]
\end{definition} 

\noindent However, this is not enough to formally regain standard
equations like $\langle f, g \rangle \circ h \approx \langle f
\circ h, g \circ h \rangle$. We need to also require that
exposures `extensionally preserve' products.

\begin{definition}
  A cartesian exposure $Q : \mathcal{B} \expo{} \mathcal{C}$ of a
  cartesian P-category $\mathcal{B}$ in a cartesian P-category
  $\mathcal{C}$ is \emph{product-preserving} whenever the
  canonical arrows \begin{align*}
    \langle Q\pi_1, Q\pi_2 \rangle &:
      Q(A \times B) \rightarrow QA \times QB \\
    {!}_{Q\mathbf{1}}
	&: Q\mathbf{1} \rightarrow \mathbf{1}
  \end{align*} are P-isomorphisms. We write $m_{A, B} : QA \times
  QB \xrightarrow{\cong} Q(A \times B)$ and $m_{0} : \mathbf{1}
  \xrightarrow{\cong} Q\mathbf{1}$ for their inverses.
\end{definition}

Amongst the exposures, then, the ones that are both cartesian
\emph{and} product-preserving are the ones that behave reasonably
well in interaction with the product structure. For example, it is
an easy calculation to show that 

\begin{proposition}
  In the above setting, $m_{A, B} \circ \langle Qf, Qg \rangle
  \sim Q\langle f, g \rangle$.
\end{proposition}

\noindent We can now prove that $\langle f, g \rangle \circ h
\approx \langle f \circ h, g \circ h \rangle$: \begin{align*}
  Q(\langle f, g \rangle \circ h)
    &\sim
      Q(\langle \pi_1 \circ \langle f, g \rangle \circ h, 
	      	\pi_2 \circ \langle f, g \rangle \circ h \rangle) \\
    &\sim
      m \circ 
	\langle Q(\pi \circ \langle f, g \rangle \circ h),
		Q(\pi' \circ \langle f, g \rangle \circ h) \rangle \\
    &\sim
      m \circ \langle Q(f \circ h), Q(g \circ h) \rangle
    \sim
      Q(\langle f \circ h, g \circ h \rangle)
\end{align*}

\subsection{Evaluators, Quotation Devices, and Comonadic Exposures}

Using transformations of exposures, we may begin to reinterpret
concepts from the modality-as-intension interpretation. Throughout
this section, we fix a cartesian P-category $\mathfrak{B}$, and a
cartesian, product-preserving endoexposure $Q : \mathfrak{B}
\expo{} \mathfrak{B}$.

\begin{definition}
  An \emph{evaluator} is a transformation of exposures $\epsilon :
  Q \natexp{} \textsf{Id}_{\mathfrak{B}}$.
\end{definition}

\noindent What about quoting? Given a point $a : \mathbf{1}
\rightarrow A$, its \emph{quote} is defined to be the point $Q(a)
\circ m_0 : \mathbf{1} \rightarrow QA$. We will require the
following definition:

\begin{definition}
  A arrow $\delta : QA \rightarrow Q^2 A$ is a \emph{reasonable
  quoting device} just if for any $a : \mathbf{1} \rightarrow QA$
  the following diagram commutes up to $\sim$: \[
    \begin{tikzcd}
    \mathbf{1}
      \arrow[r, "a"]
      \arrow[d, "m_0", swap]
    & Q A 
      \arrow[d, "\delta_A"] \\
    Q \mathbf{1}
      \arrow[r, "Qa", swap]
    & Q^2 A
    \end{tikzcd}
  \]
\end{definition}

\noindent A special case of this condition is the equation that
holds if a natural transformation of a similar type to $\delta$ is
monoidal, namely $\delta_\mathbf{1} \circ m_0 \sim Q(m_0) \circ
m_0$.

\begin{definition} 
  A \emph{quoter} is a transformation of exposures $\delta : Q
  \natexp Q^2$ such that every component $\delta_A : QA
  \rightarrow Q^2 A$ is a reasonable quoting device.
\end{definition}

\noindent These ingredients finally combine to form a
\emph{comonadic} exposure.

\begin{definition}
  A \emph{comonadic exposure} $(Q, \epsilon, \delta)$
  consists of an endoexposure $Q: (\mathfrak{B}, \sim)
  \looparrowright (\mathfrak{B}, \sim)$, an evaluator $\epsilon :
  Q \natexp{} \mathsf{Id}_\mathfrak{B}$, and a quoter $\delta : Q
  \natexp{} Q^2$, such that the following diagrams commute up to
  $\sim$: \[
    \begin{tikzcd}
      QA
        \arrow[r, "\delta_A"]
        \arrow[d, "\delta_A", swap]
      & Q^2 A 
        \arrow[d, "\delta_{Q(A)}"] \\
      Q^2 A 
        \arrow[r, "Q(\delta_A)", swap] 
      & Q^3 A
    \end{tikzcd} \quad
    \begin{tikzcd}
      QA
        \arrow[r, "\delta_A"]
        \arrow[d, "\delta_A", swap]
        \arrow[rd, "id_A", swap]
      & Q^2 A 
        \arrow[d, "\epsilon_{QA}"] \\
      Q^2 A
        \arrow[r, "Q(\epsilon_A)", swap]
      & QA
    \end{tikzcd}
  \]
\end{definition}

\section{Exposures \& Intensional Recursion}
  \label{sec:intrec}

Armed with the above, we can now speak of both extensional and
intensional recursion. Lawvere \cite{Lawvere2006} famously proved
a theorem which guarantees that, under certain assumptions which
we will discuss in \S\ref{sec:lawvere}, there exist fixed points
of the following sort.

\begin{definition}
  An \emph{extensional fixed point (EFP)} of an arrow $t : Y
  \rightarrow Y$ is a point $y : \mathbf{1} \rightarrow Y$ such
  that $t \circ y \sim y$. If, for a given object $Y$, every
  arrow $t : Y \rightarrow Y$ has a EFP, then we say that
  \emph{$Y$ has EFPs}.
\end{definition}

In Lawvere's paper EFPs are a kind of fixed point that
\emph{oughtn't} exist. In fact, Lawvere shows that---were truth
definable---the arrow $\lnot : \mathbf{2} \rightarrow \mathbf{2}$
representing negation would have a fixed point, i.e. a formula
$\phi$ with $\lnot\phi \leftrightarrow \phi$ that leads to
inconsistency.

EFPs do not encompass fixed points that \emph{ought} to exist. For
example, the \emph{diagonal lemma} for \emph{Peano Arithmetic}
(henceforth \textsf{PA}) stipulates that for any predicate
$\phi(x)$, there exists a closed formula $\mathbf{fix}(\phi)$ such
that \[
  \textsf{PA} \vdash
    \mathbf{fix}(\phi) \leftrightarrow \phi(\ulcorner
  \mathbf{fix}(\phi) \urcorner)
\] The formula $\mathbf{fix}(\phi)$ occurs asymmetrically: on the
left hand side of the bi-implication it appears as a truth value,
but on the right hand side it appears under a \emph{G\"odel
numbering}, i.e. an assignment $\ulcorner \cdot \urcorner$ of a
numeral to each term and formula of \textsf{PA}.  Since exposures
map values to their encoding, the following notion encompasses
this kind of `asymmetric' fixed point.

\begin{definition}
  Let $Q : \mathfrak{B} \expo{} \mathfrak{B}$ be a cartesian,
  product-preserving endoexposure.  An \emph{intensional fixed
  point (IFP)} of a arrow $t : QY \rightarrow Y$ is a point $y :
  \mathbf{1} \rightarrow Y$ such that \[
    y \sim t \circ Q(y) \circ m_0
  \] An object $A$ \emph{has IFPs (w.r.t. $Q$)} if every arrow $t
  : QA \rightarrow A$ has a IFP.
\end{definition}

This makes intuitive sense: $y : \mathbf{1} \rightarrow Y$ is
extensionally equal to $t$ `evaluated' at the point $Q(y) \circ
m_0 : \mathbf{1} \rightarrow QY$, which is the `quoted' version of
$y$. 

\section{Consistency, Truth and Provability: G\"odel and Tarski}
  \label{sec:consistency}

We are now in a position to argue that two well-known theorems
from logic can be reduced to very simple algebraic arguments
involving exposures. In fact, the gist of both arguments relies on
the existence of IFPs for an `object of truth values' in a
P-category.  The theorems in question are G\"odel's \emph{First
Incompleteness Theorem} and Tarski's \emph{Undefinability Theorem}
\cite{Smullyan1992,Boolos1994}.
 
Suppose that we have some sort of object $\mathbf{2}$ of `truth
values.' This need not be fancy: we require that it has two points
$\top : \mathbf{1} \rightarrow \mathbf{2}$ and $\bot : \mathbf{1}
\rightarrow \mathbf{2}$, standing for true and false respectively.
We  also require an arrow $\lnot : \mathbf{2} \rightarrow
\mathbf{2}$ which satisfies $\lnot \circ \top \sim \bot$ and
$\lnot \circ \bot \sim \top$.

A simplified version of G\"odel's First Incompleteness theorem for
\textsf{PA} is this:

\begin{theorem}[G\"odel]
  If \textsf{PA} is consistent, then there are sentences $\phi$ of
  $\textsf{PA}$ such that neither $\textsf{PA} \vdash \phi$ nor
  $\textsf{PA} \vdash \lnot \phi$.
\end{theorem}

The proof relies on two constructions: the diagonal lemma,
and the fact that provability is definable in the system. The
definability of provability amounts to the fact that there is a
formula $\text{Prov}(x)$ with one free variable $x$ such that
$\textsf{PA} \vdash \phi$ if and only if $\textsf{PA} \vdash
\text{Prov}(\ulcorner \phi \urcorner)$. That is: the system can
internally talk about its own provability, modulo some G\"odel
numbering.

It is not then hard to sketch the proof to G\"odel's theorem: take
$\psi$ such that $\textsf{PA} \vdash \psi \leftrightarrow \lnot
\text{Prov}(\ulcorner \psi \urcorner)$. Then $\psi$ is provable if
and only if it is not, so if either $\textsf{PA} \vdash \psi$ or
$\textsf{PA} \vdash \lnot\psi$ we would observe inconsistency.
Thus, if $\textsf{PA}$ is consistent, neither $\psi$ nor its
negation are provable. It follows that $\psi$ is neither
equivalent to $\bot$ or to $\top$. In a way, $\psi$ has an eerie
truth value, neither $\top$ nor  $\bot$.

Let us represent the provability predicate as an arrow $p :
Q\mathbf{2} \rightarrow \mathbf{2}$ such that $y \sim \top$ if and
only if $p \circ Q(y) \circ m_0 \sim \top$. Consistency is
captured by the following definition: \begin{definition}
  An object $\mathbf{2}$ as above is \emph{simply consistent} just
  if $\top \not\sim \bot$.
\end{definition}

Armed with this machinery, we can transport the argument
underlying G\"odel's proof to our more abstract setting:
\begin{theorem}
  If a $p : Q\mathbf{2} \rightarrow \mathbf{2}$ is as
  above, and $\mathbf{2}$ has IFPs, then one of the following
  things is true: either (a) there are points of $\mathbf{2}$
  other than $\top : \mathbf{1} \rightarrow \mathbf{2}$ and $\bot
  : \mathbf{1} \rightarrow \mathbf{2}$; or (b) $\mathbf{2}$ is
  \emph{not} simply consistent, i.e. $\top \sim \bot$.
\end{theorem}

\begin{proof}
  As $\mathbf{2}$ has IFPs, take $y : \mathbf{1} \rightarrow
  \mathbf{2}$ such that $y \sim \lnot \circ p \circ Q(y) \circ
  m_0$. Now, if $y \sim \top$, then by the property of $p$ above,
  $p \circ Q(y) \circ m_0 \sim \top$, hence $\lnot \circ p \circ
  Q(y) \circ m_0 \sim \bot$, hence $y \sim \bot$. So either $y
  \not\sim \top$ or $\mathbf{2}$ is not simply consistent.
  Similarly, either $y \not\sim \bot$ or $\mathbf{2}$ is not
  simply consistent.
\end{proof}

Tarski's \emph{Undefinability Theorem}, on the other hand is the
result that \emph{truth cannot be defined in arithmetic}
\cite{Smullyan1992}.

\begin{theorem}[Tarski]
  If \textsf{PA} is consistent, then there is no predicate
  $\text{True}(x)$ such that $\textsf{PA} \vdash \phi
  \leftrightarrow \text{True}(\ulcorner \phi \urcorner)$ for all
  sentences $\phi$.
\end{theorem}

The proof is simple: use the diagonal lemma to obtain a closed
$\psi$ such that $\textsf{PA} \vdash \psi \leftrightarrow \lnot
\text{True}(\ulcorner \psi \urcorner)$, so that $\textsf{PA}
\vdash \psi \leftrightarrow \lnot \psi$, which leads to
inconsistency.

Now, a proof predicate would constitute an evaluator $\epsilon : Q
\natexp{} \textsf{Id}_\mathfrak{B}$: we would have that \[
        \epsilon_\mathbf{2} \circ Q(y) \circ m_0
  \sim  y \circ \epsilon_\mathbf{1} \circ m_0
  \sim  y
\] where the last equality is because $\mathbf{1}$ is terminal.
This is actually a more general \begin{lemma}
  Let $Q : \mathfrak{B} \expo{} \mathfrak{B}$ be an endoexposure,
  and let $\epsilon : Q \natexp{} \textsf{Id}_\mathfrak{B}$ be an
  evaluator. Then, if $A$ has IFPs then it also has EFPs.
\end{lemma}

\begin{proof}
  Given $t : A \rightarrow A$, consider $t \circ \epsilon_A : QA
  \rightarrow A$. A IFP for this arrow is a point $y : \mathbf{1}
  \rightarrow A$ such that $y \sim t \circ \epsilon_A \circ Q(y)
  \circ m_0$. But we may calculate as above to show that
  $\epsilon_A \circ Q(y) \circ m_0 \sim y$ and thus $y \sim t
  \circ y$.
\end{proof}

In proving Tarski's theorem, we constructed a sentence $\psi$ such
that $\textsf{PA} \vdash \psi \leftrightarrow \lnot \psi$. This
can be captured abstractly by the following definition.
\begin{definition}
  An object $\mathbf{2}$ as above is \emph{fix-consistent} just if
  the arrow $\lnot : \mathbf{2} \rightarrow \mathbf{2}$ has no
  EFP; that is, there is no $y : \mathbf{1} \rightarrow
  \mathbf{2}$ such that $\lnot \circ y \sim y$. 
\end{definition}

Putting these together, we get

\begin{theorem}
  If $\mathbf{2}$ has IFPs in the presence of an evaluator, then
  it is \emph{not} fix-consistent.
\end{theorem}

\section{An Exposure on Arithmetic}
  \label{sec:arith}

We will substantiate the results of the previous section by
sketching the construction of a P-category and endoexposure based
on a first-order theory. The method is very similar to that of
Lawvere \cite{Lawvere2006}, and we will also call it the
\emph{Lindenbaum P-category} of the theory. The construction is
general, and so is the thesis of this section: an exposure on a
Lindenbaum P-category abstractly captures the notion of a
well-behaved G\"odel `numbering' on the underlying theory.

Let there be a single-sorted first-order theory $\mathsf{T}$.  The
objects of the P-category are the formal products of (a)
$\mathbf{1}$, the terminal object, (b) $A$, the
\emph{domain}, and (c) $\mathbf{2}$, the object of \emph{truth
values}. Arrows $\mathbf{1} \rightarrow A$ and $A \rightarrow A$
are \emph{terms} with no or one free variable respectively. Arrows
$A^n \rightarrow \mathbf{2}$ and $\mathbf{1} \rightarrow
\mathbf{2}$ are \emph{predicates}, with $n$ and no free variables
respectively. Finally, arrows $\mathbf{2}^n \rightarrow
\mathbf{2}$ can be thought of as \emph{logical connectives} (e.g.
$\land : \mathbf{2} \times \mathbf{2} \rightarrow \mathbf{2}$).

Two arrows $s, t : C \rightarrow A$ with codomain $A$ (i.e. two
terms of the theory) are related if and only if they are provably
equal, i.e. $s \sim t$ iff $\textsf{T} \vdash s = t$. Two arrows
$\phi, \psi : C \rightarrow \mathbf{2} $ with codomain
$\mathbf{2}$ are related if and only if they are provably
equivalent, i.e. $\phi \sim \psi$ iff $\textsf{T} \vdash \phi
\leftrightarrow \psi$.

To define an exposure, it suffices to have a G\"odel numbering,
i.e. a representation of terms and formulas of the theory as
elements of its domain $A$. More precisely, we need a G\"odel
numbering for which \emph{substitution is internally definable}.
We write $\ulcorner \phi(x_1, \dots, x_n) \urcorner$ and
$\ulcorner t(a_1, \dots, a_m) \urcorner$ for the G\"odel numbers
of the formula $\phi(x_1, \dots, x_n)$ and the term $t(a_1, \dots,
a_m)$ respectively, and we assume that $\ulcorner \cdot \urcorner$
is injective.  Let $QA \myeq A$, $Q(\mathbf{2}) \myeq A$, and
$Q(\mathbf{1}) \myeq \mathbf{1}$.  Finally, define $Q$ to act
component-wise on finite products: this will guarantee that it is
cartesian and product-preserving.

The action on arrows is what necessitated that substitution be
definable: this amounts to the existence of a term $sub(y, x)$
with the property that if $\phi(x)$ is a predicate and $t$ is a
term, then $\textsf{T} \vdash sub(\ulcorner
\phi \urcorner, \ulcorner t \urcorner) = \ulcorner \phi(t)
\urcorner$.  Now, given a predicate $\phi : A \rightarrow
\mathbf{2}$ with one free variable, $Q(\phi) : A \rightarrow A$ is
defined to be the term $sub(\ulcorner \phi \urcorner, x)$. Given a
sentence $\phi : \mathbf{1} \rightarrow \mathbf{2}$, we define
$Q(\phi) : \mathbf{1} \rightarrow A$ to be exactly the closed term
$\ulcorner \phi \urcorner$. The action is similar on arrows with
codomain $A$, and component-wise on product arrows. The last axiom
of exposures is satisfied: if $Q\phi \sim Q\psi$, then $\ulcorner
\phi \urcorner = \ulcorner \psi \urcorner$, so that $\phi = \psi$,
by the injectivity of the G\"odel numbering.

In this setting, IFPs really are fixpoints of formulas.

\section{Where do IFPs come from?}
  \label{sec:lawvere}

In \S\ref{sec:intrec} we mentioned Lawvere's fixed point theorem.
This theorem guarantees the existence of EFPs under the assumption
that there is an arrow of this form:

\begin{definition}
  An arrow $r : X \times A \rightarrow Y$ is \emph{weakly-point
  surjective} if, for every $f : A \rightarrow Y$, there exists a
  $x_f : \mathbf{1} \rightarrow X$ such that for all points $a :
  \mathbf{1} \rightarrow A$ it is the case that $r \circ \langle
  x_f, a \rangle \sim f \circ a$.
\end{definition}

So a weak-point surjection is a bit like `pointwise cartesian
closure,' in that the effect of all arrows $A \rightarrow Y$ on
points $\mathbf{1} \rightarrow A$ is representable by some point
$\mathbf{1} \rightarrow X$, w.r.t.  $r$. Lawvere noticed that if
the `exponential' $X$ and the domain $A$ coincide, then a simple
diagonal argument yields fixpoints for all arrows $Y \rightarrow
Y$.

\begin{theorem}[Lawvere]
  If $r : A \times A \rightarrow Y$ is a weak-point surjection,
  then every arrow $t : Y \rightarrow Y$ has an extensional fixed
  point (EFP).
\end{theorem}
\begin{proof}
  Let $f \myeq t \circ r \circ \langle id_A, id_A \rangle$. Then
  there exists a $x_f : \mathbf{1} \rightarrow A$ such that $r
  \circ \langle x_f, a \rangle \sim f \circ a$ for all $a :
  \mathbf{1} \rightarrow A$. We compute that $r \circ \langle x_f,
  x_f \rangle \sim\ t \circ r \circ \langle id_A, id_A \rangle
  \circ x_f \sim\ t \circ r \circ \langle x_f, x_f \rangle$, so
  that $r \circ \langle x_f, x_f \rangle$ is a EFP of $t$.
\end{proof}

Can we adapt Lawvere's result to IFPs? The answer is positive, and
rather straightforward once we embellish the statement with
appropriate occurrences of $Q$. We also need a reasonable quoting
device.

\begin{theorem}
  Let $Q$ be a monoidal exposure, and let $\delta_A : QA
  \rightarrow Q^2 A$ be a reasonable quoting device. If $r : QA
  \times QA \rightarrow Y$ is a weak-point surjection then every
  arrow $t : QY \rightarrow Y$ has an intensional fixed point.
\end{theorem}
\begin{proof}
  Let $f \myeq t \circ Qr \circ m_{QA, QA} \circ
  \langle \delta_A, \delta_A \rangle$. Then there exists a $x_f :
  \mathbf{1} \rightarrow QA$ such
  that $r \circ \langle x_f, a \rangle \sim f \circ a$ for all $a
  : \mathbf{1} \rightarrow QA$. We compute that
  \begin{align*}
    r \circ \langle x_f, x_f \rangle
    \sim\  & t \circ Qr \circ m \circ \langle \delta_A, \delta_A \rangle \circ x_f
    \sim\  t \circ Qr \circ m \circ \langle \delta_A \circ x_f, \delta_A \circ x_f \rangle \\
    \sim\  & t \circ Qr \circ m \circ \langle Q(x_f) \circ m_0, Q(x_f) \circ m_0 \rangle \\
    \sim\  & t \circ Qr \circ m \circ \langle Q(x_f), Q(x_f) \rangle \circ m_0 \\
    \sim\  & t \circ Qr \circ Q(\langle x_f, x_f \rangle) \circ m_0
    \sim\  t \circ Q(r \circ \langle x_f, x_f \rangle) \circ m_0
  \end{align*} so that $r \circ \langle x_f, x_f \rangle$ is a IFP
  of $t$.
\end{proof}

In the next section, we shall see that this is a true categorical
analogue of Kleene's \emph{Second Recursion Theorem (SRT)}.

\section{The Recursion Theorems}
  \label{sec:recthm}

In fact, the theorem we just proved in \S\ref{sec:lawvere} is
strongly reminiscent of a known theorem in (higher order)
computability theory, namely a version of Kleene's \emph{First
Recursion Theorem (FRT)}. 

Let us fix some notation.  We write $\simeq$ for \emph{Kleene
equality}: we write $e \simeq e'$ to mean either that both
expressions $e$ and $e'$ are undefined, if either both are
undefined, or both are defined and of equal value. Let $\phi_0$,
$\phi_1$, $\dots$ be an enumeration of the partial recursive
functions. We will also require the \emph{s-m-n theorem} from
computability theory. Full definitions and statements may be found
in the book by Cutland \cite{Cutland1980}.

\begin{theorem}[First Recursion Theorem]
  \label{thm:frt}
  Let $\mathcal{PR}$ be the set of unary partial recursive
  functions, and let $F : \mathcal{PR} \rightarrow \mathcal{PR}$
  be an \emph{effective operation}. Then $F : \mathcal{PR}
  \rightarrow \mathcal{PR}$ has a fixed point.
\end{theorem}

\begin{proof}
  That $F : \mathcal{PR} \rightarrow \mathcal{PR}$ is an effective
  operation means that there is a partial recursive $f :
  \mathbb{N} \times \mathbb{N} \parfunc{} \mathbb{N}$ such that
  $f(e, x) \simeq F(\phi_e)(x)$. Let $d \in \mathbb{N}$ a code for
  the partial recursive function $\phi_d(y, x) \myeq f(S(y, y),
  x)$, where $S : \mathbb{N} \times \mathbb{N} \parfunc{}
  \mathbb{N}$ is the s-1-1 function of the s-m-n theorem. Then, by
  the s-m-n theorem, and the definitions of $d \in \mathbb{N}$ and
  $f$, \[
    \phi_{S(d, d)}(x)
      \simeq \phi_d(d, x)
      \simeq f(S(d, d), x)
      \simeq F(\phi_{S(d,d)})(x)
  \] so that $\phi_{S(d, d)}$ is a fixed point of $F :
  \mathcal{PR} \rightarrow \mathcal{PR}$.
\end{proof}

Lawvere's theorem is virtually identical to a point-free version
of this proof. Yet, one cannot avoid noticing that we have proved
more than that for which we bargained. The $f : \mathbb{N} \times
\mathbb{N} \parfunc{} \mathbb{N}$ in the proof above has the
special property that it is \emph{extensional}, in the sense that
\[
  \phi_e = \phi_{e'} \quad\Longrightarrow\quad
    \forall x \in \mathbb{N}.\ f(e, x) = f(e', x)
\] However, the step which yields the fixed point argument holds
for any such $f$, not just the extensional ones. This fact
predates the FRT, and was shown by Kleene in 1938
\cite{Kleene1938}.

\begin{theorem}[Second Recursion Theorem]
  For any partial recursive $f: \mathbb{N} \times \mathbb{N}
  \parfunc \mathbb{N}$, there exists $e \in \mathbb{N}$ such that
  $\phi_e(y) \simeq f(e, y)$ for all $y \in \mathbb{N}$.
\end{theorem} 

This is significantly more powerful than the FRT, as $f(e, y)$ can
make arbitrary decisions depending on the source code $e$,
irrespective of the function $\phi_e$ of which it is the source
code. Moreover, it is evident that the function $\phi_e$ has
\emph{access to its own code}, allowing for a certain degree of
reflection. Even if $f$ is extensional, hence defining an
effective operation, the SRT grants us more power than the FRT:
for example, before recursively calling $e$ on some points, $f(e,
y)$ could `optimise' $e$ depending on what $y$ is, hence ensuring
that the recursive call will run faster than $e$ itself would.
This line of thought is common in the \emph{partial evaluation}
community, see e.g. \cite{Jones1996}.

In the sequel we argue that our fixed point theorem involving
exposures is a generalisation of Lawvere's theorem, in the exact
same way that the SRT is a non-extensional generalisation of the
FRT. In order to do so, we define a P-category and an exposure
based on realizability theory, and claim that the FRT and the SRT
are instances of the general theorems in that particular
P-category.


\section{An Exposure on Assemblies}
  \label{sec:asm}

Our second example of an exposure will come from realizability,
where the basic objects are assemblies. An assembly is a set to
every element of which we have associated a set of
\emph{realizers}. The elements of the set can be understood as
\emph{elements of a datatype}, and the set of realizers of each
such element as the \emph{machine-level representations} of it.
For example, if realizers range in the natural numbers, then
assemblies and functions between them which are partial recursive
on the level of realizers yield a category where `everything is
computable.'

In practice, the generalisation from natural numbers to an
arbitrary \emph{partial combinatory algebra (PCA)} is made. A PCA
is an arbitrary, untyped `universe' corresponding to some notion
of computability or realizability. There are easy tricks with
which one may encode various common `first-order' datatypes, such
as booleans, integers, etc. as well as all partial recursive
functions (up to the encoding of integers). These methods can be
found \cite{Beeson1985,Longley1995,Longley1997,vanOosten2008}.

\begin{definition} A \emph{partial combinatory algebra (PCA)} $(A,
  \cdot)$ consists of a set $A$, its carrier, and a partial binary
  operation $\cdot : A \times A \parfunc A$ such that there exist
  $\mb{K}, \mb{S} \in A$ with the properties that \[
    \mb{K} \cdot x\downarrow, 
      \quad
    \mb{K} \cdot x \cdot y \simeq y,
      \quad
    \mb{S} \cdot x \cdot y\downarrow,
      \quad
    \mb{S} \cdot x \cdot y \cdot z \simeq x \cdot z \cdot (y \cdot z)
  \] for all $x, y, z \in A$.
\end{definition}

The simplest example of a PCA, corresponding to classical
computability, is $K_1$, also known as \emph{Kleene's first
model}. Its carrier is $\mathbb{N}$, and $r \cdot a \myeq
\phi_r(a)$.

\begin{definition} An \emph{assembly} $X$ on a PCA $A$ consists of a set
$\left\lvert X \right\rvert$ and, for each $x \in \left\lvert X
\right\rvert$, a non-empty subset $\left\|x\right\|_X$ of $A$.  If
$a \in \left\|x\right\|_X$, we say that \emph{$a$ realizes $x$}.
\end{definition}

\begin{definition} For two assemblies $X$ and $Y$, a function $f :
\bars{X} \rightarrow \bars{Y}$ is said to be \emph{tracked by $r
\in A$} just if, for all $x \in \bars{X}$ and $a
\in \dbars{x}_X$, we have $r \cdot a \downarrow$ and $r \cdot a
\in \dbars{f(x)}_Y$ \end{definition}

Now: for each PCA $A$, we can define a category $\mathbf{Asm}(A)$,
with objects all assemblies $X$ on $A$, and morphisms $f : X
\rightarrow Y$ all functions $f : \bars{X} \rightarrow \bars{Y}$
that are tracked by some $r \in A$.
\begin{theorem}
  Assemblies and `trackable' morphisms between them form a
  category $\mathbf{Asm}(A)$ that is cartesian closed, has
  finite coproducts, and a natural numbers object.
\end{theorem}

We only mention one other construction that we shall need.
Given an assembly $X$, the \emph{lifted assembly} $X_\bot$ is
defined to be \[
  \bars{X_\bot} \myeq \bars{X} \cup \{ \bot \}
    \quad \text{and} \quad
  \dbars{x}_{X_\bot} \myeq 
    \begin{cases}
      \setcomp{ r }{ r \cdot \overline{0}\downarrow \text{
	and } r \cdot \overline{0} \in
	\dbars{x}_X }
	  &\text{for } x \in \bars{X} \\
      \setcomp{ r }{r \cdot \overline{0}\uparrow}
	  &\text{for } x = \bot
    \end{cases}
\] for some chosen element of the PCA $\overline{0}$.
Elements of $X_\bot$ are either elements of $X$, or the undefined
value $\bot$. Realizers of $x \in \bars{X}$ are `computations' $r
\in A$ which, when run (i.e. given the dummy value $\overline{0}$
as argument) return a realizer of $x$. A computation that does not
halt when run represents the undefined value.\footnote{Bear in
mind that this definition of the lifted assembly does not work if
the PCA is total. We are mostly interested in the decidedly
non-total PCA $K_1$, so this is not an issue. There are other,
more involved ways of defining the lifted assembly; see
\cite{Longley1997} in particular.}

\subsection{Passing to a P-category}

The lack of intensionality in the category $\mathbf{Asm}(A)$ is
blatantly obvious. To elevate a function $f : \bars{X} \rightarrow
\bars{Y}$ to a morphism $f : X \rightarrow Y$, we only require
that \emph{there exists} a `witness' $r \in A$ that realizes it,
and then we forget about this witness entirely. To mend this, we
define a P-category.

The P-category $\mathfrak{Asm}(A)$ of assemblies on $A$ is
defined to have all assemblies $X$ on $A$ as objects, and pairs 
$(f : \bars{X} \rightarrow \bars{Y}, r \in A)$ where $r$ tracks
$f$ as arrows. We define $(f, r) \sim (g, s)$ just if $f = g$,
i.e. when the underlying function is the same. The composition of
$(f, p) : X \rightarrow Y$ and $(g, q) : Y \rightarrow Z$ is $(g
\circ f, \mathbf{B} \cdot q \cdot p)$ where $\mathbf{B}$ is a
combinator in the PCA such that $\mathbf{B} \cdot f \cdot g \cdot
x \simeq f \cdot (g \cdot x)$ for any $f, g, x \in A$. The
\emph{identity} $id_X : X \rightarrow X$ is defined to be
$(id_{\bars{X}}, \mathbf{I}) : X \rightarrow X$, where
$\mathbf{I}$ is a combinator in the PCA such that $\mathbf{I}
\cdot x \simeq x$ for all $x \in A$.

Much in the same way as before---but now up to the PER $\sim$---we
can show \begin{theorem}
  $\mathfrak{Asm}(A)$ is a cartesian closed P-category with a
  natural numbers object $\mathbb{N}$.
\end{theorem} 

We can now define an exposure $\Box : \mathfrak{Asm}(A)
\expo{} \mathfrak{Asm}(A)$. For an assembly $X \in
\mathfrak{Asm}(A)$, let $\Box X$ be the assembly defined by \[
  \bars{\Box X} \myeq
    \setcomp{(x, a)}{x \in \bars{X}, a \in \dbars{x}_A},
      \quad
    \dbars{(x, a)}_{\Box X} \myeq \{\ a\ \}
\] Given $(f, r) : X \rightarrow Y$, we define $\Box(f, r) = (f_r,
r) : \Box X \rightarrow \Box Y $ where $f_r : \bars{\Box X}
\rightarrow \bars{\Box Y}$ is defined by $f_r(x, a) \myeq (f(x), r
\cdot a)$. Thus, under the exposure each element $(x, a) \in
\bars{\Box X}$ carries with it its own unique realizer $a$. The
image of $(f, r)$ under $\Box$ shows not only what $f$ does to an
element of its domain, but also how $r$ acts on the realizer of
that element.

It is long but straightforward to check that \begin{theorem} 
  $\Box : \mathfrak{Asm}(A) \expo{} \mathfrak{Asm}(A)$ is a
  cartesian, product-preserving, and comonadic endoexposure.
\end{theorem}

\subsection{Kleene's Recursion Theorems, categorically}
  \label{sec:krtcat}

Let us concentrate on the category $\mathfrak{Asm}(K_1)$. Arrows
$\mathbb{N} \rightarrow \mathbb{N}_\bot$ are easily seen to
correspond to partial recursive functions. It is not hard to
produce a weak-point surjection $r_E : \mathbb{N} \times
\mathbb{N} \rightarrow \mathbb{N}_\bot^{\mathbb{N}}$, and hence to
invoke Lawvere's theorem to show that every arrow
$\mathbb{N}_\bot^{\mathbb{N}} \rightarrow
\mathbb{N}_\bot^{\mathbb{N}}$ has an extensional fixed point. Now,
by Longley's generalised Myhill-Shepherdson theorem
\cite{Longley1995,Longley2015}, arrows
$\mathbb{N}_\bot^{\mathbb{N}} \rightarrow
\mathbb{N}_\bot^{\mathbb{N}}$ correspond to effective operations.
Hence, in this context Lawvere's theorem corresponds to the simple
diagonal argument that we used to show the FRT.\footnote{But note
that this is not the complete story, as there is no guarantee that
the fixed point obtained in \emph{least}, which is what Kleene's
original proof in \cite{Kleene1952} gives. See also
\cite{Kavvos2016c}.}

Let us look at arrows of type $\Box (\mathbb{N}_\bot^{\mathbb{N}})
\rightarrow \mathbb{N}_\bot^{\mathbb{N}}$. These correspond to
`non-functional' transformations, mapping functions to functions,
but without respecting extensionality. As every natural number
indexes a partial recursive function, these arrows really
correspond to all partial recursive functions (up to some tagging
and encoding). It is not hard to see that $\Box \mathbb{N}$ is
P-isomorphic to $\mathbb{N}$, and that one can build a weak-point
surjection of type $\Box \mathbb{N} \times \Box \mathbb{N}
\rightarrow \mathbb{N}_\bot^{\mathbb{N}}$, so that by our theorem,
every arrow of type $\Box (\mathbb{N}_\bot^{\mathbb{N}})
\rightarrow \mathbb{N}_\bot^{\mathbb{N}}$ has an intensional fixed
point. This is exactly Kleene's SRT!

\section{Rice's theorem}
  \label{sec:rice}

To further illustrate the applicability of the language of
exposures, we state and prove an abstract version of \emph{Rice's
theorem}. Rice's theorem is a result in computability which states
that no computer can decide any non-trivial property of a program
by looking at its code. A short proof relies on the SRT.

\begin{theorem}[Rice]
  Let $\mathcal{F}$ be a non-trivial set of partial recursive
  functions, and let $A_\mathcal{F} \myeq \setcomp{ e \in
  \mathbb{N} }{ \phi_e \in \mathcal{F}}$ be the set of indices of
  functions in that set. Then $A_\mathcal{F}$ is undecidable.
\end{theorem}
\begin{proof} Suppose $A_\mathcal{F}$ is decidable. The fact
$\mathcal{F}$ is non-trivial means that there is some $a \in
\mathbb{N}$ such that $\phi_a \in \mathcal{F}$ and some $b \in
\mathbb{N}$ such that $\phi_b \not\in\mathcal{F}$. Consequently, $a
\in A_\mathcal{F}$ and $b \not\in A_\mathcal{F}$.

Define $f(e, x) \simeq \textbf{if } e \in A_\mathcal{F} \textbf{
then } \phi_b(x) \textbf{ else } \phi_a(x)$.  By Church's thesis,
$f : \mathbb{N} \times \mathbb{N} \rightarrow \mathbb{N}$ is
partial recursive. Use the SRT to obtain $e \in \mathbb{N}$ such
that $\phi_e(x) \simeq f(e, x)$.  Now, either $e \in
A_\mathcal{F}$ or not. If it is, $\phi_e(x) \simeq f(e, x) \simeq
\phi_b(x)$, so that $\phi_e \not\in \mathcal{F}$, a contradiction.
A similar phenomenon occurs if $e \not\in A_\mathcal{F}$.
\end{proof}

Constructing the function $f$ in the proof required three basic
elements: (a) the ability to evaluate either $\phi_a$ or $\phi_b$
given $a$ and $b$; (b) the ability to decide which one to use
depending on the input; and (c) intensional recursion. For (a), we
shall need evaluators, for (b) we shall need that the truth object
$\mathbf{2}$ is a \emph{weak coproduct} of two copies of
$\mathbf{1}$, and for (c) we shall require IFPs.

\begin{theorem}
  Let $\mathbf{2}$ is a simply consistent `truth object' which
  also happens to be a a weak coproduct of two copies of
  $\mathbf{1}$, with injections $\top : \mathbf{1} \rightarrow
  \mathbf{2}$ and $\bot : \mathbf{1} \rightarrow \mathbf{2}$.
  Furthermore, suppose that $A$ has EFPs. If $f : A \rightarrow
  \mathbf{2}$ is such that for all $x : \mathbf{1} \rightarrow A$,
  either $f \circ x \sim \top$ or $f \circ x \sim \bot$, then
  $f$ is trivial, in the sense that either $f \circ x \sim \top$
  for all $x : \mathbf{1} \rightarrow A$, or $f \circ x \sim
  \bot$ for all $x : \mathbf{1} \rightarrow A$.
\end{theorem}
\begin{proof}
  Suppose there are two such distinct $a, b : \mathbf{1}
  \rightarrow A$ such that $f \circ a \sim \top$ and $f \circ b
  \sim \bot$. Let $g \myeq [b, a] \circ f$ and let $y : \mathbf{1}
  \rightarrow A$ be its EFP. Now, either $f \circ y \sim \top$ or
  $f \circ y \sim \bot$. In the first case, we can calculate that
  $\top \sim\ f \circ [b, a] \circ f \circ y \sim\ f \circ [b, a]
  \circ \top \sim\ f \circ b \sim\ \bot$ so that $\mathbf{2}$ is
  \emph{not simply consistent}. A similar situation occurs if $f
  \circ y \sim \bot$.
\end{proof}

Needless to say that the premises of this theorem are easily
satisfied in our exposure on assemblies from \S\ref{sec:asm}
if we take $A = {\mathbb{N}_\bot}^\mathbb{N}$ and $\mathbf{2}$ to
be the lifted coproduct $(\mathbf{1} + \mathbf{1})_\bot$.

\section{Conclusion}

We have modelled intensionality with P-categories, and introduced
a new construct that abstractly corresponds to G\"odel numbers.
This led us to an immediate unification of many `diagonal
arguments' in logic and computability, as well as a new
perspective on the notion of \emph{intensional recursion}. Our
approach is clearer and more systematic than the one in
\cite{Lawvere2006}.

Many questions are left open. We are currently working on the
medium-term goal of a safe, reflective programming language based
on modal type theory. The basics are there, but there are many
questions: what operations should be available at modal
types; with how much expressivity would the language be endowed
for each possible set; and what are the applications?

On the more technical side, it is interesting to note that we have
refrained from a categorical proof of the diagonal lemma for
\textsf{PA}. All our attempts were inelegant, and we believe that
this is because arithmetic is fundamentally untyped:
$Q(\mathbf{2})$ has many more points than `all G\"odel numbers of
predicates.' In contrast, our approach using exposures is typed,
which sets it apart from all previous attempts at capturing such
arguments categorically, including the very elegant work of
Cockett and Hofstra \cite{Cockett2008,Cockett2010}. The approach
in \emph{op.  cit.} is based on \emph{Turing categories}, in
which every object is a retract of some very special objects---the
\emph{Turing objects}. In the conclusion of \cite{Cockett2008}
this is explicitly mentioned as an `inherent limitation.' Only
time will tell which approach is more encompassing.

Finally, it would be interesting to study the meaning of exposure
in examples not originating in logic and computability, but in
other parts of mathematics. Can we find examples of exposures
elsewhere? Are they of any use?

\section*{Acknowledgements}

I would like to thank my doctoral supervisor, Samson Abramsky, for
suggesting the topic of this paper, and for his help in
understanding the issues around intensionality and intensional
recursion.

\bibliographystyle{plain}


\begin{thebibliography}{10}
\providecommand{\url}[1]{\texttt{#1}}
\providecommand{\urlprefix}{URL }

\bibitem{Abramsky2014}
Abramsky, S.: {Intensionality, Definability and Computation}. In: Baltag, A.,
  Smets, S. (eds.) Johan van Benthem on Logic and Information Dynamics, pp.
  121--142. Springer-Verlag (2014),
  \url{https://dx.doi.org/10.1007/978-3-319-06025-5_5}

\bibitem{Barendregt1991}
Barendregt, H.: {Self-Interpretation in Lambda Calculus}. Journal of Functional
  Programming  1(2),  229--233 (1991),
  \url{https://dx.doi.org/10.1017/S0956796800020062}

\bibitem{Beeson1985}
Beeson, M.J.: {Foundations of Constructive Mathematics}. Springer Berlin
  Heidelberg (1985), \url{https://dx.doi.org/10.1007/978-3-642-68952-9}

\bibitem{Bierman2000a}
Bierman, G.M., de~Paiva, V.: {On an Intuitionistic Modal Logic}. Studia Logica
  65(3),  383--416 (2000), \url{https://dx.doi.org/10.1023/A:1005291931660}

\bibitem{Boolos1994}
Boolos, G.S.: {The Logic of Provability}. Cambridge University Press, Cambridge
  (1994), \url{https://dx.doi.org/10.1017/CBO9780511625183}

\bibitem{Cockett2008}
Cockett, J.R.B., Hofstra, P.J.W.: {Introduction to Turing categories}. Annals
  of Pure and Applied Logic  156(2-3),  183--209 (2008),
  \url{http://dx.doi.org/10.1016/j.apal.2008.04.005}

\bibitem{Cockett2010}
Cockett, J.R.B., Hofstra, P.J.W.: {Categorical simulations}. Journal of Pure
  and Applied Algebra  214(10),  1835--1853 (2010),
  \url{http://dx.doi.org/10.1016/j.jpaa.2009.12.028}

\bibitem{Cubric1998}
{\v{C}}ubri{\'{c}}, D., Dybjer, P., Scott, P.J.: {Normalization and the Yoneda
  embedding}. Mathematical Structures in Computer Science  8(2),  153--192
  (1998), \url{https://dx.doi.org/10.1017/s0960129597002508}

\bibitem{Cutland1980}
Cutland, N.: {Computability: An Introduction to Recursive Function Theory}.
  Cambridge University Press (1980)

\bibitem{Davies2001a}
Davies, R., Pfenning, F.: {A modal analysis of staged computation}. Journal of
  the ACM  48(3),  555--604 (2001),
  \url{http://dl.acm.org/citation.cfm?id=382785}

\bibitem{Fitting2015}
Fitting, M.: {Intensional Logic}. In: Zalta, E.N. (ed.) The Stanford
  Encyclopedia of Philosophy. Metaphysics Research Lab, Stanford University,
  summer 201 edn. (2015),
  \url{https://plato.stanford.edu/archives/sum2015/entries/logic-intensional/}

\bibitem{Jones1996}
Jones, N.D.: {An introduction to partial evaluation}. ACM Computing Surveys
  28(3),  480--503 (1996), \url{http://doi.acm.org/10.1145/243439.243447}

\bibitem{Kavvos2016c}
Kavvos, G.A.: {Kleene's Two Kinds of Recursion}. CoRR  abs/1602.0 (2016),
  \url{http://arxiv.org/abs/1602.06220}

\bibitem{Kleene1938}
Kleene, S.C.: {On notation for ordinal numbers}. The Journal of Symbolic Logic
  3(04),  150--155 (1938), \url{https://dx.doi.org/10.2307/2267778}

\bibitem{Kleene1952}
Kleene, S.C.: {Introduction to Metamathematics}. North-Holland, Amsterdam
  (1952)

\bibitem{Lawvere2006}
Lawvere, F.W.: {Diagonal arguments and cartesian closed categories}. Reprints
  in Theory and Applications of Categories  15,  1--13 (2006),
  \url{http://www.tac.mta.ca/tac/reprints/articles/15/tr15abs.html}

\bibitem{Longley1995}
Longley, J.R.: {Realizability Toposes and Language Semantics}. Ph.D. thesis,
  University of Edinburgh. College of Science and Engineering. School of
  Informatics. (1995),
  \url{http://www.lfcs.inf.ed.ac.uk/reports/95/ECS-LFCS-95-332/}

\bibitem{Longley2000}
Longley, J.R.: {Notions of computability at higher types I}. In: Logic
  Colloquium 2000: Proceedings of the Annual European Summer Meeting of the
  Association for Symbolic Logic, held in Paris, France, July 23-31, 2000.
  Lecture Notes in Logic, vol.~19, pp. 32--142. A. K. Peters (2005)

\bibitem{Longley2015}
Longley, J.R., Normann, D.: {Higher-Order Computability}. Theory and
  Applications of Computability, Springer Berlin Heidelberg, Berlin, Heidelberg
  (2015), \url{https://dx.doi.org/10.1007/978-3-662-47992-6}

\bibitem{Longley1997}
Longley, J.R., Simpson, A.K.: {A uniform approach to domain theory in
  realizability models}. Mathematical Structures in Computer Science  7,
  469--505 (1997), \url{https://dx.doi.org/10.1017/S0960129597002387}

\bibitem{vanOosten2008}
van Oosten, J.: {Realizability: An Introduction to its Categorical Side}, vol.
  152. Elsevier (2008),
  \url{http://www.sciencedirect.com/science/bookseries/0049237X/152}

\bibitem{Davies2001}
Pfenning, F., Davies, R.: {A judgmental reconstruction of modal logic}.
  Mathematical Structures in Computer Science  11(4),  511--540 (2001),
  \url{https://dx.doi.org/10.1017/S0960129501003322}

\bibitem{Smith1984}
Smith, B.C.: {Reflection and Semantics in LISP}. In: Proceedings of the 11th
  ACM SIGACT-SIGPLAN Symposium on Principles of Programming Languages (POPL
  '84). pp. 23--35. ACM Press, New York, New York, USA (1984),
  \url{https://dx.doi.org/10.1145/800017.800513}

\bibitem{Smullyan1992}
Smullyan, R.M.: {G{\"{o}}del's Incompleteness Theorems}. Oxford University
  Press (1992)

\bibitem{Wand1998}
Wand, M.: {The Theory of Fexprs is Trivial}. LISP and Symbolic Computation
  10(3),  189--199 (1998), \url{https://dx.doi.org/10.1023/A:1007720632734}

\end{thebibliography}


\end{document}